\documentclass[letterpaper, 10 pt, conference, onecolumn]{ieeeconf}

\IEEEoverridecommandlockouts       
\usepackage{cite}	     	
\usepackage{textcomp}   	
\usepackage[euler]{textgreek}
\usepackage{lipsum}	     	
\usepackage{soul, color}   	
\usepackage[svgnames]{xcolor} 
\usepackage[utf8]{inputenc} 
\usepackage[T1]{fontenc}    
\let\labelindent\relax  	
\usepackage[shortlabels]{enumitem}  		
\usepackage[normalem]{ulem}
\usepackage{eurosym}

\usepackage{amsmath}    	
\usepackage{amssymb}		
\usepackage{amsfonts}		

\usepackage{amsthm}    	    
\usepackage{mathtools}		
\usepackage{newtxmath} 	    
\usepackage{cases}     	
\usepackage{bm}             

\usepackage[caption=false,font=footnotesize]{subfig}
\usepackage{graphicx}		
\usepackage{booktabs}   	
\usepackage{siunitx}    
\usepackage[export]{adjustbox}  
\usepackage{floatrow}		

\usepackage{hyperref}		

\graphicspath{{figures/}}   
\setlist{noitemsep} 		
\hypersetup{			
	colorlinks,
	linkcolor={DodgerBlue},
	citecolor={DodgerBlue},
	urlcolor={DodgerBlue}
}
\allowdisplaybreaks

\newcommand{\T}{^{\mathsf{T}}}

\newcommand{\B}[1]{\if#1\relax\bm{#1}\else\mathbf{#1}\fi} 
\newcommand{\R}[1]{\mathrm{#1}}						      
\newcommand{\C}[1]{\mathcal{#1}}	
\newcommand{\BB}[1]{\mathbb{#1}}

\newcommand{\norm}[1]{\left\lVert #1 \right\rVert}

\newtheoremstyle{mythmstyle}
{1ex}
{1ex}
{\itshape}
{}
{\bfseries}
{.}
{ }
{\thmname{#1}\thmnumber{ #2}\thmnote{ (#3)}}  

\newtheoremstyle{myremark} 
{}
{}
{}
{}
{\bfseries}
{.}
{.5em}
{}

\theoremstyle{mythmstyle}			

\newtheorem{proposition}{Proposition}

\newtheorem{lemma}{Lemma}

\theoremstyle{myremark} 

\title{Intermittent non-pharmaceutical strategies to mitigate the COVID-19 epidemic in a network model of Italy via constrained optimization}

\author{%
Marco Coraggio\textsuperscript{*, 1},
Shihao Xie\textsuperscript{*, 2}, 
Francesco De Lellis\textsuperscript{1}, 
Giovanni Russo\textsuperscript{\#, 3}, 
Mario di Bernardo\textsuperscript{\#, 1}%
\thanks{\textsuperscript{*, \#}These authors contributed equally.}
\thanks{\textsuperscript{1}Department of Electrical Engineering and Information Technology, University of Naples Federico II, Italy (email: \{marco.coraggio, francesco.delellis, mario.dibernardo\}@unina.it).}
\thanks{\textsuperscript{2}School of Electrical and Electronic Engineering, University College Dublin (email: shihao.xie1@ucdconnect.ie).}
\thanks{\textsuperscript{3}Department of Information and Electric Engineering and Applied Mathematics, University of Salerno (email: giovarusso@unisa.it).}
\thanks{This work was supported by the Research
Project PRIN 2017 ``Advanced Network Control of Future Smart Grids'' funded by the Italian Ministry of University and Research (2020–2023)--http://vectors.dieti.unina.it.
The work of Shihao Xie was supported by the Science Foundation Ireland (SFI) under Grant 16/IA/4610.}%
}

\begin{document}

\maketitle
\thispagestyle{empty}
\pagestyle{empty}


\begin{abstract}
This paper is concerned with the design of intermittent non-pharmaceutical strategies to mitigate the spread of the COVID-19 epidemic exploiting network epidemiological models.
Specifically, by studying a variational equation for the dynamics of the infected, we derive, using contractivity arguments, a condition that can be used to guarantee that the effective reproduction number is less than unity.
This condition (i) is easily computable, (ii) is interpretable, being directly related to the model parameters, and (iii) can be used to enforce a scalability condition that prohibits the amplification of disturbances within the network system. 
We then include satisfaction of such a condition as a constraint in a Model Predictive Control problem so as to mitigate (or suppress) the spread of the epidemic while minimizing the economic impact of the interventions. 
A data-driven model of Italy as a network of three macro-regions (North, Center, and South), whose parameters are  identified from real data, is used to illustrate and evaluate the effectiveness of the proposed control strategy.
\end{abstract}



\section{Introduction}%
\label{sec:introduction}%
Because of recent events, the design of control strategies aimed at mitigating the spread of epidemic processes has suddenly become a major pressing  research problem with clear societal and economic impacts \cite{bin2020postlockdown, Kissler860}. 
Among these strategies, non-pharmaceutical interventions (NPIs) are essential to mitigate or suppress the epidemic (see \cite{imperial_report9,brauner2021inferring} for further details on the definition and difference between mitigation and suppression).
In this context, two key challenges are: (i) devising proxies to detect the onset of the epidemic spreading; (ii) synthesizing control strategies that balance the need of containing the epidemic with the need of reducing its economic impact.

\subsubsection*{Related work}

A fast intermittent lockdown strategy was proposed in \cite{bin2020postlockdown} with the goal of alleviating the effects of measurement uncertainty. 
Following a similar argument, a cyclic periodic lockdown strategy was studied in \cite{Karin2020.04.04.20053579}, whereas a stepping down policy was presented in \cite{KENNEDY2020104440}.
The use of feedback to dynamically adapt the intervention strategies was recently suggested in several papers.
For example, in \cite{9099929}, several feedback intervention strategies are presented, whereas in \cite{dellarossa2020network}
the benefits were shown of adopting feedback intermittent regional control strategies triggered by occupancy levels of the intensive care units (ICUs) in each region. 
Furthermore, several epidemic control strategies were recently proposed, based on model predictive control (MPC). 
These include \cite{morato2020optimal-1}, where MPC is used to decide whether to enforce social distancing in a region/country of interest, using a SIRASDC model ($S$: susceptible, $I$: infected, $R$: recovered, $A$: asymptomatic, $S$: symptomatic, $D$: dead, $C$: with control) and including hard constraints on ICUs occupancy and minimal dwell time of the control. 
MPC was also implemented on a number of different models, e.g., a SIDARTHE model 
($S$: susceptible, $I$: infected, $D$: diagnosed, $A$: ailing, $R$: recognized, $T$: threatened, $H$: healed, $E$: extinct)
in \cite{kohler2020robust}, a SEAIR model 
($S$: susceptible, $E$: exposed, $A$: asymptomatic infected, $I$: confirmed infections, $R$: recovered)
in \cite{Tsay}, and a multi-region SIRQTHE network model 
($S$: susceptible, $I$: infected, $R$: recovered, $Q$: quarantined, $T$: threatened, $H$: healed, $E$: extinct) 
in \cite{CARLI2020373}.

\subsubsection*{Contribution}
We consider network epidemic models, and give a sufficient condition, directly dependent on the model parameters, ensuring that the epidemiological \emph{effective reproduction number} $\mathcal{R}_t$ is less than one and that the infected population decreases to zero.
Then, we consider a network model of Italy parametrized on real data, and obtain a set of intermittent non-pharmaceutical intervention strategies at the regional level by formulating a constrained optimization problem that is effectively solved via MPC.
The formulation (i) embeds our novel stability condition as a constraint, (ii) is based on the minimization of an economic cost induced by both social distancing and testing and (iii) considers a lower bound on the control dwell time.
To the best of our knowledge, we are the first to implement the MPC on a regional model including all these features.

\section*{Mathematical Preliminaries}
We denote by $\C{I}$ the identity matrix of appropriate dimensions, by $\B{0}$ and $\B{1}$ column vectors of appropriate dimensions consisting of all zeros and ones, respectively; $\left\lVert \cdot \right\rVert_q$ denotes the vector $\ell_q$-norm and its induced matrix norm.
$\BB{R}_{\succeq 0}^n$ is the positive orthant in $\BB{R}^n$ (including the axes); $\BB{N}_{\ge 0}$ and $\BB{N}_{> 0}$ are the sets of semipositive and positive natural numbers, respectively.
Letting $\B{A}$ be a matrix, $\B{A} \succeq 0$ indicates that all its elements $A_{ij} \ge 0$ (analogously for $\succ$, $\preceq$, $\prec$).
\begin{lemma}[\cite{horn2012matrix}, p. 368]\label{lem:bounds_matrices_norms} 
Let $\B{B}_1, \B{B}_2 \in \BB{R}^{n \times n}$, with $n \in \BB{N}_{>0}$.
If $\B{B}_1 \succeq 0$, $\B{B}_2 \succeq 0$, and $\B{B}_1 \preceq \B{B}_2$, then $\norm{\B{B}_1} \le \norm{\B{B}_2}$.
\end{lemma}

\section{Model description}%
\label{sec:model_description}%

To model the dynamics of the epidemic, we adapted and updated a network model that was originally presented and validated in \cite{dellarossa2020network}. 
We re-parameterized and updated the model with real data related to the spread of the epidemic in Italy collected by the National Civil Protection Authority \cite{protezione_civile}, between February 24, 2020 and February 25, 2021.
The model comprises $M = 3$ macro-regions corresponding
to the North, the Center, and the South of Italy.%
\footnote{The North comprises the regions Piedmont, Lombardy, Liguria, Veneto, Friuli-Venezia Giulia, Trentino-Alto Adige, Emilia Romagna; the Center comprises Tuscany, Umbria, Marche, Lazio, Abruzzo; the South comprises Molise, Campania, Basilicata, Apulia, Calabria, Sicily, Sardinia.}


The dynamics of the disease in each region $i \in \{ 1, \ldots, M\}$ is captured by a SIQHDR model, describing the discrete-time evolution of the number of susceptible ($S_i$), (undetected) infected ($I_i$), quarantined ($Q_i$), hospitalized ($H_i$), deceased ($D_i$), and recovered ($R_i$) as a function of time $t \in \BB{N}_{\ge 0}$ (in days).
We also define
$\B{S} \coloneqq [S_1 \ \cdots \ S_M]\T$,
$\B{I} \coloneqq [I_1 \ \cdots \ I_M]\T$,
$\B{Q} \coloneqq [Q_1 \ \cdots \ Q_M]\T$,
$\B{H} \coloneqq [H_1 \ \cdots \ H_M]\T$,
$\B{D} \coloneqq [D_1 \ \cdots \ D_M]\T$,
$\B{R} \coloneqq [R_1 \ \cdots \ R_M]\T$, 
and let $\B{x} \coloneqq [\B{S}\T \ \B{I}\T \ \B{Q}\T \ \B{H}\T \ \B{D}\T \ \B{R}\T ]\T \in \BB{R}^{6M}_{\succeq 0}$.
Moreover, for each pair $(i,j)$ of macro-regions, $\phi_{ij}(t) \in [0,1]$ describes the fraction of people in region $i$ that commute and interact with people in region $j$, returning to their region of origin at the end of each day $t$ (see also \cite{stolerman2015sirnetwork}).
The resulting network model is 
\begin{subequations}\label{eq:network_model_discrete}
\begin{align}
{S}_i(t+1) &= S_i(t) - \beta S_i(t) \sum_{j=1}^M \frac{\rho_j(t) \phi_{ij}(t)}{N_j^{\R{p}}(t)} {\sum_{k=1}^M{ \phi_{kj}(t)I_k(t)}},\\
{I}_i(t+1)&= I_i(t) + \beta S_i(t) \sum_{j=1}^M \frac{\rho_j(t) \phi_{ij}(t)}{N_j^{\R{p}}(t)} {\sum_{k=1}^M{ \phi_{kj}(t)I_k(t)}} \nonumber \\
&\phantom{={}} -(\gamma +\alpha_i(t) + \psi_i)I_i(t), \label{eq:network_model_discrete_I} \\
{Q}_i(t+1) &= Q_i(t) + \alpha_i(t) I_i(t) - (\kappa_i^H + \eta_i^Q)Q_i(t) + \kappa_i^Q H_i(t),\\
H_i(t+1) &= H_i(t) + \kappa_{i}^{H} Q_{i}(t) + \psi_{i} I_{i}(t) \nonumber \\
&\phantom{={}}- \left( \eta_{i}^{H} + \kappa_{i}^{Q} + \zeta \left( H_{i}(t) \right) \right) H_{i}(t), \\
{D}_i(t+1)&= D_i(t) + \zeta \left(H_{i}(t) \right)  H_{i}(t),\\
{R}_i(t+1)&= R_i(t) + \gamma I_i(t) + \eta_i^Q Q_i(t) + \eta_i^H H_i(t),
\end{align}
\end{subequations}
where
\begin{equation}\label{eq:N_i_p}
	N_i^\R{p}(t) = \sum_{k = 1}^M \phi_{ki}(t) \ (S_k(t) + I_k(t) + R_k(t)),
\end{equation}
and the initial conditions are given at time $t_0 \in \BB{N}_{\ge 0}$ as $\B{x}(t_0) = \B{x}^0 \in \BB{R}_{\succeq 0}^{6M}$.

Following \cite{dellarossa2020network}, in \eqref{eq:network_model_discrete} the infection rate $\beta$ and the recovery rate $\gamma$ are taken to be equal for all the macro-regions, as they have not been shown to strongly depend on regional factors.
The other parameters vary among different macro-regions to capture the heterogeneity of containment strategies, different regional health care systems, etc.
In particular, note that
\begin{equation}\label{eq:bound_parameters_I_out}
    \gamma + \alpha_i(t) + \psi_i \le 1, \quad \forall i \in \{1, \dots, M\}, \ \forall t \in \BB{N}_{\ge 0},
\end{equation}
as in \eqref{eq:network_model_discrete_I} the number of people leaving compartment $I_i$ can never be greater than $I_i$. 
Moreover, as done in \cite{dellarossa2020network}, we also estimate the number of patients requiring treatment in intensive care as 10\% of those who are hospitalized.
Thus, letting $T_i^H$ be the maximum number of beds in the ICUs of region $i$ \cite{dati_salute_gov, linee_indirizzo}, following \cite{dellarossa2020network} we express the mortality rate as
$\zeta(H_i(t)) = \zeta^0 + \zeta^{\R{b}} \min \left\{ \frac{0.1 H_i(t)}{T_i^H}, 1 \right\}$.
Finally, $N_i^\R{p}(t)$ is the population in region $i$ that is free-to-move at time $t$, whereas $N_i$ is the total population in region $i$.
All parameters' values are identified from data reported in \cite{protezione_civile} and are contained in Table \ref{tab:model_parameter_values}.

\begin{table}[t]
\begin{center}
\caption{Model parameter values; for a description of all parameters see \cite{dellarossa2020network}, page 6. 
The values of $T_i^H$ are reported from \cite{ICU_2021}} 
\label{tab:model_parameter_values}
\begin{tabular}{@{}l|l|lll@{}}
\toprule
Parameter  & Global &  North & Center & South \\
\midrule
$\beta$ & $0.4$ & - & - & - \\
$\gamma$ & $0.07$ & - & - & -  \\
$\zeta^0$ & $0.0168$ & - & - & - \\
$\zeta^{\R{b}}$ & $0.0068$ & - & - & - \\
$\psi_i$ & - & $0.0327$ & $0.0922$ & $0.1042$\\
$\eta_i^H$ & - & $0.6086$ & $0.4439$ & $0.7565$\\ 
$\eta_i^Q$ & - & $0.0141$ & $0.0063$ & $0.0075$ \\
$\kappa_i^H$ & - & $0.0375$ & $0.0200$ & $0.0429$ \\
$\kappa_i^Q$ & - & $0.0344$ & $0.0202$ & $0.0179$\\
$\phi_{1j}^0$ & - & $0.9979$ & $0.0013$ & $0.0008$ \\
$\phi_{2j}^0$ & - & $0.0030$ & $0.9949$ & $0.0021$ \\
$\phi_{3j}^0$ & - & $0.0011$ & $0.0024$ & $0.9965$ \\
$T_i^H$ & - & $4660$ & $2775$ & $3170$ \\
\bottomrule
\end{tabular}
\end{center}
\end{table}

\subsubsection*{Social distancing}

The parameter $\rho_i(t) \in [0,1]$ models the effect of social distancing in region $i$ at time $t$.
Namely, $\rho_i(t) = 1$ corresponds to the absence of social distancing measures, whereas $\rho_i(t) = 0$ means that contact between people is forbidden.
Intermediate values indicate different levels of social distancing.

\subsubsection*{Travel restrictions}

In what follows we assume
\begin{equation*}
    \phi_{ij}(t) =
    \begin{dcases}
        \varphi_{i}(t) \varphi_{j}(t) \phi_{ij}^0, & i \ne j, \\
    1 - \sum_{j = 1, j \ne i}^M \phi_{ij}(t), & \text{otherwise},
    \end{dcases}
\end{equation*}
where $\varphi_{i}(t) \in [0, 1]$ is a parameter modelling the strength of travel restrictions to/from region $i$ on day $t$, and $\phi_{ij}^0$ is the influx from region $i$ to region $j$ in the absence of travel restrictions.
In particular, $\varphi_{i}(t) = 1$ means that travel to/from region $i$ is possible, whereas $\varphi_{i}(t) = 0$ implies it is forbidden.

\subsubsection*{Testing}

The effect of testing-and-tracing is captured by the parameter $\alpha_i(t)$ in \eqref{eq:network_model_discrete}, which models the rate at which undetected infected $I_i$ are moved to the quarantined compartment $Q_i$ at time $t$. 
To describe the ability of a region to modulate the amount of testing carried out within its territory, we set $\alpha_i(t) = \alpha^0 + \sigma_i(t) \tilde{\alpha}$,
where  $\sigma_i(t) \in [0, 1]$ is a parameter denoting the testing capacity of region $i$ on day $t$.
In particular, testing is set to the minimal level, $\alpha^0 \in \BB{R}_{> 0}$, when $\sigma_i(t) = 0$, or to its maximal value, $\alpha^0+\tilde{\alpha}$, when $\sigma_i(t) = 1$. 

The values $\alpha^0$ and $\tilde{\alpha}$ were identified  as $\alpha^0 = 0.0671$ and $\tilde{\alpha} = 0.0806$ by performing a linear regression of the data collected by each Italian region in the time window from April $22$, $2020$ to February $25$, $2021$, with a resulting correlation coefficient of $0.312$, as reported in Figure \ref{fig:regression}.
\begin{figure}[h!]
  \centering
  \includegraphics[width=0.7\textwidth]{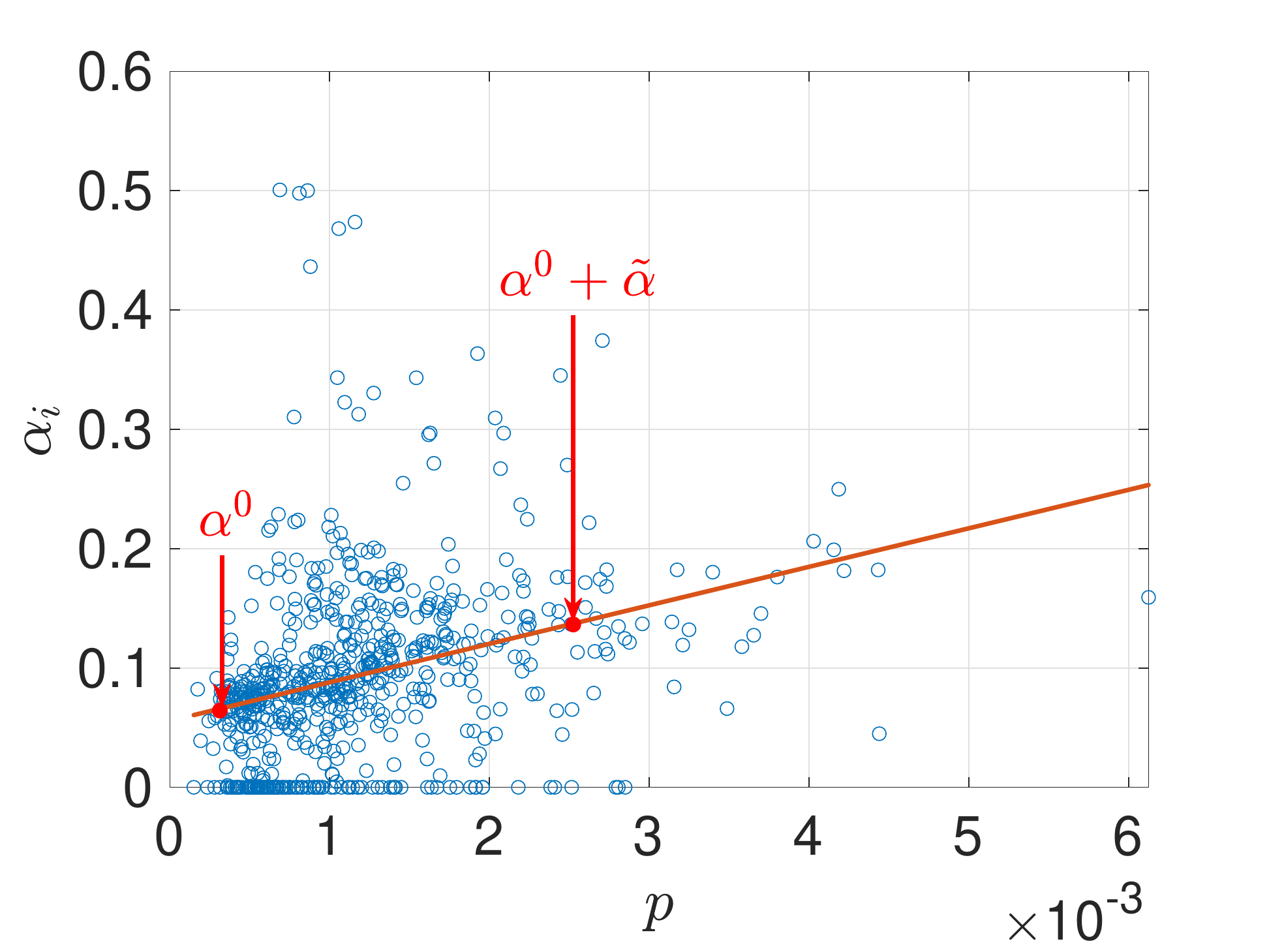}
  \caption{Linear regression of the values of $\alpha_i$ in each region (see (\ref{eq:network_model_discrete})) over the fraction $p$ of population tested daily between April $22, 2020$ and February $25, 2021$, taken from \cite{protezione_civile}.}
  \label{fig:regression}
\end{figure}



\subsection*{Reproduction number}
A crucial quantity to assess the rate of propagation of an epidemic disease is the \emph{reproduction number}, defined in epidemiology as ``the number of secondary infections generated by one infected person'' \cite[Ch.~27]{Rothman}.
Typically, the \emph{basic} reproduction number $\C{R}_0$ is computed when the entire population is assumed to be susceptible \cite{Heffernan}, while the \emph{effective} reproduction number $\C{R}_t$ is estimated from data \cite{Rothman, Nishiura}. Specifically, the effective reproduction number of region $i$, denoted by $\C{R}_{t,i}$, can be estimated via the method described in \cite[§8]{kohlberg2020demystifying}, that is, by computing
\begin{equation}\label{eq:estimation_R_t_i}
    \C{R}_{t,i} = \frac{\sum_{k = t - 3}^t [(N_i - S_i(k)) - (N_i - S_i(k-1))]}{\sum_{k = t - 7}^{t-4} [(N_i - S_i(k)) - (N_i - S_i(k-1))]}.
\end{equation}

Note that while an analytical expression of $\C{R}_0$ for network epidemic models can be found in the existing literature, e.g., \cite{stolerman2015sirnetwork}, \cite{diekmann1990definition}, at the moment an algebraic expression for $\C{R}_t$ does not exist in a closed form for network models.
Later in Section \ref{sec:theretical_section}, we derive a sufficient condition to ensure that $\C{R}_t$ (for the network) is below the critical unitary threshold.


\section{Analysis}%
\label{sec:theretical_section}

We let $\B{\rho}(t) \coloneqq [\rho_1(t) \ \cdots \ \rho_M(t)]\T$,
$\B{\varphi}(t) \coloneqq [\varphi_1(t) \ \cdots \ \varphi_M(t)]\T$,
$\B{\sigma}(t) \coloneqq [\sigma_1(t) \ \cdots \ \sigma_M(t)]\T$,
and define the vector of time-varying parameters  $\B{u}(t) \coloneqq [\B{\rho}\T(t) \ \B{\varphi}\T(t) \ \B{\sigma}\T(t)]\T \in [0, 1]^{3M}$, which will be used later in Sections \ref{sec:control_problem} and \ref{sec:mpc_implementation_and_validation} as the vector of control inputs to mitigate the epidemic.

\subsection{Dynamics and variation of the infected around a trajectory}

We focus on the dynamics of the infected recasting \eqref{eq:network_model_discrete_I} in a more compact form as
\begin{equation}\label{eq:network_model_discrete_I_compact}
    \B{I}(t + 1) = 
    (\C{I} + \B{\Psi}(\B{x}(t), \B{u}(t))) \  \B{I}(t),
\end{equation}
where we stressed the dependency of $\B{\Psi}$ on the state vector $\B{x}$ and the control input $\B{u}$, and where
\begin{equation}\label{eq:Psi_ij}
	\Psi_{ij}(\B{x}, \B{u}) := 
	\begin{dcases}
		\beta S_i \sum_{k=1}^M \frac{ \rho_k \phi_{ik} \phi_{jk}}{N_k^\R{p}} - (\alpha_i + \psi_i + \gamma), & i = j, \\
		\beta S_i \sum_{k=1}^M \frac{ \rho_k \phi_{ik} \phi_{jk}}{N_k^\R{p}}, & i \ne j.
	\end{dcases}
\end{equation}

In what follows, say $\bar{\B{x}}(t) \coloneqq [\bar{\B{S}}\T(t) \ \bar{\B{I}}\T(t) \ \bar{\B{Q}}\T(t) \ \bar{\B{H}}\T(t) \ \bar{\B{D}}\T(t) \ \bar{\B{R}}\T(t) ]\T$
the state evolution of \eqref{eq:network_model_discrete} starting from some initial condition $\bar{\B{x}}(t_0) = \bar{\B{x}}^0 \in \BB{R}^{6M}_{\succeq 0}$, $t_0\in\BB{N}_{\ge 0}$, when $\B{u}(t)$ is equal to some $\bar{\B{u}}(t)$.

Additionally, following the epidemiological definition of $\C{R}_t$ we obtain an expression that yields the variations of $\B{I}(t)$, say $\delta {\B{I}}(t)$, around $\bar{\B{I}}(t)$, due to small perturbations in the infected only.
Formally, this can be done by considering in \eqref{eq:network_model_discrete_I_compact} all the states ($\B{x}(t)$) except $\B{I}(t)$ as inputs to the dynamics of the infected compartment.
After some algebraic manipulations, this dynamics can be given as%
\footnote{Note that \eqref{eq:variation_dynamics_I} only contains the variations of $\B{I}(t)$ due to $\B{I}(t)$ itself, consistently with the epidemiological interpretation of the reproduction number.}
\begin{equation}\label{eq:variation_dynamics_I}
    \delta\B{I}(t+1) 
	= (\C{I} + \B{\Psi}(\bar{\B{x}}(t), \bar{\B{u}}(t)) + \B{\Gamma}(\bar{\B{x}}(t), \bar{\B{u}}(t))) \ \delta\B{I}(t),
\end{equation}
where
\begin{equation}\label{eq:Gamma_ij}
	\Gamma_{ij}(\B{x}, \B{u}) := \sum_{q = 1}^M \frac{\partial \Psi_{iq}}{\partial I_j} I_q = 
	- \beta S_i \sum_{k=1}^M \left( \frac{\rho_k \phi_{ik} \phi_{jk}}{N_k^\R{p}} \sum_{q=1}^M \frac{\phi_{qk}}{N_k^\R{p}} I_q \right).
\end{equation}

\subsection{A sufficient condition for epidemic containment}

Next, we use \eqref{eq:network_model_discrete_I_compact} and \eqref{eq:variation_dynamics_I} to give a simple yet effective condition to (i) enforce suppression of the epidemic and (ii) to ensure that the number of secondary infections generated by perturbing the number of infected along a trajectory decays exponentially to zero.

\begin{proposition}\label{pro:stability_trajectory}
    Assume that there exists some matrix norm, say $\norm{\cdot}_q$, and some constant $c < 1$ such that
    \begin{equation}\label{eq:stability_condition}
        \left\lVert \C{I} + \B{\Psi}(\bar{\B{x}}(t), \bar{\B{u}}(t)) \right\rVert_q \le c<1, \quad \forall t \ge t_0, \ t_0 \in \BB{N}_{\ge 0}.
    \end{equation} 
    Then, (i) $\norm{\bar{\B{I}}(t)}_q \le c^{t-t_0} \norm{\bar{\B{I}}(t_0)}_q$ for all $t \ge t_0$, and (ii) $\norm{ \delta \B{I}(t)}_q \le c^{t-t_0} \norm{\delta\B{I}(t_0)}_q$ for all $t \ge t_0$.
\end{proposition}
\begin{proof}
For the sake of brevity, we denote $\B{\Psi}(\bar{\B{x}}(t), \bar{\B{u}}(t))$ by $\B{\Psi}(t)$ and $\B{\Gamma}(\bar{\B{x}}(t), \bar{\B{u}}(t))$ by $\B{\Gamma}(t)$.

Part (i) follows directly from the fact that, from \eqref{eq:network_model_discrete_I_compact}, we have
  $\norm{\bar{\B{I}}(t)}_q = \norm{\prod_{\tau=t_0}^{t-1} \left(\C{I} + \B{\Psi}(\tau) \right) \ \bar{\B{I}}(t_0)}_q \le
  \prod_{\tau = t_0}^{t-1} \norm{\C{I} + \B{\Psi}(\tau)}_q  \norm{\bar{\B{I}}(t_0)}_q$.
   
Analogously, to prove part (ii), it suffices to exploit \eqref{eq:variation_dynamics_I} and show that \eqref{eq:stability_condition} implies
\begin{equation}\label{eq:stability_condition_variations}
        \left\lVert \C{I} + \B{\Psi}(t) + \B{\Gamma}(t) \right\rVert_q \le c < 1, \quad
        \forall t \ge t_0, \ t_0 \in \BB{N}_{\ge 0}.
\end{equation}
To show this, we let 
$\B{B}_1(t) \coloneqq \C{I} + \B{\Psi}(t) + \B{\Gamma}(t)$ and
$\B{B}_2(t) \coloneqq \C{I} + \B{\Psi}(t)$, and note the following facts:
(a) recalling \eqref{eq:bound_parameters_I_out} and \eqref{eq:Psi_ij}, we have $\B{B}_2(t) \succeq 0$, $\forall t \in \BB{N}_{\ge 0}$;
(b) as
$\B{\Gamma}(\B{x}, \B{u}) \preceq 0, \forall \B{x} \in \BB{R}_{\succeq 0}^{6M}, \forall \B{u} \in [0, 1]^{3M}$,
it holds that $\B{B}_1(t) \preceq \B{B}_2(t)$, $\forall t \in \BB{N}_{\ge 0}$;
(c) we can rewrite the term in \eqref{eq:Gamma_ij}, using \eqref{eq:N_i_p}, as
\begin{equation*}
\sum_{q=1}^M \frac{\phi_{qk}}{N_k^\R{p}} I_q = \frac{\sum_{q=1}^M \phi_{qk} I_q}{\sum_{q=1}^M \phi_{qk} (S_q + I_q + R_q)} \le 1, \ \forall k \in \{1, \dots, M\},
\end{equation*}
and, recalling \eqref{eq:bound_parameters_I_out}, \eqref{eq:Psi_ij}, \eqref{eq:Gamma_ij}, we get that $\B{B}_1(t) \succeq 0, \forall t \in \BB{N}_{\ge 0}$.

Finally, exploiting the facts (a), (b), (c) we can apply Lemma \ref{lem:bounds_matrices_norms} and, considering \eqref{eq:stability_condition}, we have $\norm{\B{B}_1(t)} \le \norm{\B{B}_2(t)} \le c$, $\forall t\ge t_0$, that is \eqref{eq:stability_condition_variations}, which proves (ii).
\end{proof}

Next, we give an interpretation of Proposition \ref{pro:stability_trajectory} in terms of $\C{R}_t$. 
Recall that $\C{R}_t$ is defined as the average number $b \ge 0$ of secondary infections generated by $a = 1$ infected people.
The well-known condition $\C{R}_t < 1$ means that one infected person will infect on average less than one person, and the epidemic will disappear, with the total number of infected (that is $\left\lVert \B{I}(t) \right\rVert_1 = \sum_{i = 1}^M \left| I_i(t) \right|$) decreasing at all time.
Indeed, this corresponds to part (i) of the thesis of Proposition \ref{pro:stability_trajectory}, when considered with $q = 1$.
Moreover, noting that in our framework $a$ corresponds to $\left\lVert \delta \B{I}(t) \right\rVert_1$, it is obvious that each person infects on average less than one person if and only if $\left\lVert \delta \B{I} (t) \right\rVert_1$ decreases monotonically over time, which corresponds to part (ii) of Proposition \ref{pro:stability_trajectory}.

%
%

Additionally, as recently noted in \cite[Remark 6]{shihao}, the use of the $\infty$-norm prohibits the amplification of bounded deterministic disturbance on the infected dynamics within the network system.
Its application in condition \eqref{eq:stability_condition} yields
$ 
    \max_{i \in \{1, \dots, M\}} \left| \delta I_i(t + 1) \right|  < \max_{i \in \{1, \dots, M\}} \left| \delta I_i(t)\right| 
$, 
which means that the maximum variation of infected across all regions decreases over time.
We argue that this condition might be effectively used to assess quantitatively the trend of an epidemic spread. 

As proven in the Appendix, Proposition \ref{pro:stability_trajectory} can also be related to a condition independently obtained in \cite{ma2021optimal} when specialized to a SIS network model (see the Appendix for the formal statement and further details). 

\section{Control problem formulation}%
\label{sec:control_problem}%

We now formulate our control problems, aimed at either mitigating or suppressing the epidemic, using different combinations of control actions, while minimizing the economic impact of the measures.

\subsubsection*{Suppression of the epidemic}

First, we consider the problem of finding the cheapest strategies in terms of social distancing ($\B{\rho}$) and travel restrictions ($\B{\varphi}$) that enforce suppression of the epidemic, assuming testing ($\B{\sigma}$) is not increased beyond its nominal value.
To this aim, we define---omitting time dependence for the sake of brevity,
\begin{equation}\label{eq:A_t_i}
    A_i(\B{x}, \B{u}) \coloneqq \left\lvert 1 + \Psi_{ii}(\B{x}, \B{u}) \right\rvert + \sum_{j = 1, j \ne i}^M \left\lvert \Psi_{ij}(\B{x}, \B{u}) \right\rvert.
\end{equation}
Clearly, if $A_i \le c < 1$, $\forall i, \forall t \ge t_0$, then \eqref{eq:stability_condition} in Proposition \ref{pro:stability_trajectory} holds with the infinity norm, and the infected population decreases to zero.
Thus, the control goal can be formulated as the constrained optimization problem
\begin{subequations}\label{eq:optimization_problem_suppression}
\begin{align}
 \min_{\C{U}_{[1, T-1]}} \quad 
 & J_{[1, T]}, \\
 \text{s.t.} \quad &\eqref{eq:network_model_discrete}, \quad \forall t, \nonumber \\
 & A_{i}(\B{x}(t), \B{u}(t)) \le c < 1, \quad \forall i, \ \forall t, \label{eq:optimization_problem_suppression_stability}\\
 & \rho_i(t) \in \C{S}_\rho, \ \varphi_i(t) \in \C{S}_\varphi, \ \sigma_i(t) = 0, \quad \forall i, \forall t, \nonumber \\
 &{w_\B{u} = \tau^*}, \label{eq:dwell_time}
\end{align}
\end{subequations}
where $J_{[1, T]}$ is the economic cost function over the time window $[1, T]$, and will be presented in Section \ref{sec:cost_function},
$\C{U}_{[1, T-1]} \coloneqq (\B{u}(1), \B{u}(2), \dots, \B{u}(T-1))$ is the set of all decision variables, 
$\C{S}_\rho$, $\C{S}_\varphi$ are sets of possible values of $\rho_i$ and $\varphi_i$, which will be selected in Section \ref{sec:mpc_implementation_and_validation}, and
$w_\B{u}$ is the \emph{minimum dwell time} on $\B{u}$.
In loose terms, \eqref{eq:dwell_time} means that the value of $\B{u}$ cannot be changed before at least $\tau^* \in \BB{N}_{>0}$ days have passed since the last change (similarly to what done in \cite{morato2020optimal-1, CARLI2020373}); a formal implementation of \eqref{eq:dwell_time} is detailed in Section \ref{sec:mpc_implementation_and_validation}.


\subsubsection*{Mitigation of the epidemic}

To better preserve the economy, and as currently done in many countries affected by the epidemic, we also study a different approach, where some of the restrictions can be lifted when certain conditions are met.

To this aim, we obtain a relaxation of condition \eqref{eq:stability_condition} by considering the condition
\begin{equation}\label{eq:relaxed_stability}
    \forall i,\forall t, \quad \text{if } \pi_i(t) \implies A_{i}(\B{x}(t), \B{u}(t)) \le c < 1.
\end{equation}
where $\pi_i(t)$ is a Boolean condition associated to the ``criticality'' of the epidemic in region $i$, that we choose as
\begin{equation}
    \pi_i(t) = \begin{dcases}
        1, & 0.1 H_i(t) \ge \epsilon_H T_i^H \ \vee \ \C{R}_{t, i} \ge \epsilon_\C{R}, \\
        0, & \text{otherwise},
    \end{dcases}
\end{equation}
linking it to the level of saturation of the ICUs in each region and the value of $\C{R}_{t,i}$ in each region, estimated using \eqref{eq:estimation_R_t_i}, with
$\epsilon_H \in (0, 1]$, $\epsilon_\C{R} \in \BB{R}_{>0}$.
The resulting control problem can be then formulated as \eqref{eq:optimization_problem_suppression} with \eqref{eq:optimization_problem_suppression_stability} substituted by \eqref{eq:relaxed_stability}.


\subsubsection*{The effect of testing}

Finally, we explore if and how increasing the amount of testing beyond the nominal value in each region can further reduce costs by formulating the problem as \eqref{eq:optimization_problem_suppression} where the constraint on $\sigma_i(t)$ is chosen as $\sigma_i(t) \in \C{S}_\sigma, \forall i, \forall t$, with $\C{S}_\sigma$ being a set of possible values, selected later in Section \ref{sec:mpc_implementation_and_validation}.


\subsection{Cost function}%
\label{sec:cost_function}

As a cost function for the optimization problems above, we consider the total economic cost of the measures undertaken by each region, computed accounting for the cost of people not being able to work and the cost of testing.

Let $P_{\R{active}, i}(t) \coloneqq S_i(t) + I_i(t) + R_i(t)$; $c_{\R{m}} = \text{\euro} 83.992$ be the mean daily GDP pro-capita in Italy \cite{IMF2020},
$c_{\R{w}} = 0.617$ be the fraction of people that are \emph{not} able to work when social distancing is in place \cite{holgersenwho}, and
$c_\alpha$ be the daily cost of carrying out maximal testing, normalized over the population.
Then, the economic cost at time $t$ in region $i$  is given by
$
J_i(t) = J_{i,1}(t) + J_{i,2}(t) + J_{i,3}(t) + J_{i,4}(t) + J_{i,5}(t),
$
where the addends are defined as follows:
\begin{enumerate}[1.]
    \item \label{ite:cost_1} $J_{i,1}(t)$ is the loss due to people that live and work in region $i$, but cannot work because of social distancing:
    $J_{i,1}(t) \coloneqq c_{\R{m}} c_\R{w}  \left(1 - \sqrt{\rho_i(t)}\right) \phi_{ii}^0 P_{\R{active}, i}(t)$;

    \item \label{ite:cost_2} $J_{i,2}(t)$ is the loss due to people that live in other regions and work in $i$, can travel but cannot work because of social distancing:
    $J_{i,2}(t) \coloneqq c_{\R{m}} c_\R{w}  \left( 1 - \sqrt{\rho_i(t)} \right) \sum_{j = 1, j \ne i}^N \phi_{ji}(t) P_{\R{active}, j}(t)$;
    
    \item \label{ite:cost_3} $J_{i,3}(t)$ is the cost of travel restrictions computed in terms of people that live in other regions and work in $i$, but cannot travel:
    $J_{i,3}(t) \coloneqq c_\R{m} c_{\R{w}} \sum_{j=1, j \ne i}^M \left( \phi_{ij}^0 - \phi_{ij}(t) \right) P_{\R{active}, j}(t)$;
    
    \item \label{ite:cost_4} $J_{i,4}(t)$ is the loss due to people that live in $i$ but cannot work because becoming incapacitated:
    $J_{i,4}(t) \coloneqq c_{\R{m}} [c_\R{w} Q_i(t) + H_i(t) + D_i(t)]$;
    
    \item \label{ite:cost_5} $J_{i,5}(t)$ is related to the cost of carrying out increased testing in region $i$ and is computed as
    $J_{i,5}(t) \coloneqq c_\alpha N_i \sigma_i(t)$.
\end{enumerate}

The expressions in items \ref{ite:cost_1} and \ref{ite:cost_2} were obtained by noticing that in \eqref{eq:network_model_discrete}---omitting subscripts for notational easiness---the amount of new infected at each time step is proportional to $\rho S I = (\sqrt{\rho} S) (\sqrt{\rho} I)$.
Therefore, the portion of people that is \emph{not} restrained is proportional to  $\sqrt{\rho}$, and consequently the fraction of population that \emph{is} restrained is proportional to $(1 - \sqrt{\rho})$.
In item \ref{ite:cost_5}, $c_\alpha = \tilde{p} c_\R{test} k_p = 0.3583$, where $c_\R{test} = \text{\euro} 59$ is the economic cost of a single test \cite{cicchetti2020}, $k_p = 1.74$ is the average number of tests carried out per tested person \cite{protezione_civile}, and $\tilde{p} = 3.49 \times 10^{-3}$ is the fraction of people tested daily when testing is maximal.

Letting the regional terminal cost be $\bar{J}_i(T) = 0$, and assuming the presence of a discount factor $0 < d \le 1$, the total cost over the time window $[t_0,T]$ can then be given as 
\begin{equation*} 
    J_{[t_0, T]} = \sum_{t = t_0}^{T-1} d^{(t-t_0)}\left( \sum_{i = 1}^M J_i(t) \right) + \sum_{i = 1}^M \bar{J}_i(T).
\end{equation*}

\section{MPC implementation and validation}%
\label{sec:mpc_implementation_and_validation}

We solved the control problems described in the previous section by using a Model Predictive Control approach with prediction horizon $T_\R{p}$ \cite{morari1999model}. 


We chose the sets $\C{S}_\rho$, $\C{S}_\varphi$, $\C{S}_\sigma$ of possible values of the control inputs as follows.
For $\rho_i(t)$, we approximated the values of this parameter estimated in \cite{dellarossa2020network} for each of the Italian regions when different levels of social distancing were enforced, obtaining $\C{S}_\rho = \{0.3, 0.4, 0.5, 0.6, 0.7\}$.
For $\varphi_i(t)$, we replicated the decrease of 70\% in mobility fluxes among regions as registered in \cite{google_mobility_data_2020} during the first months of the Italian national lockdown (February--May 2020), obtaining $\C{S}_\varphi = \{\sqrt{0.3}, 1\}$.
As for $\sigma_i(t)$, we chose three values to represent low, medium, and intense testing levels, i.e., $\C{S}_\sigma = \{0, 0.5, 1\}$. 

To implement the dwell time constraint \eqref{eq:dwell_time}, we prescribe that the control input $\B{u}$ on day $t$ can be changed only if $\tau(t) = \tau^*$, where $\tau(t)$ is the number of days (up to $\tau^*$), counting back from $t$ (excluding $t$), such that the control variables stayed equal to the value they had in $t-1$.
To compute $\tau(t)$, we initialize $\tau(1) = 0$, and update $\tau$ according to the law 
\begin{equation}
      \tau(t+1) = \begin{cases}
        \min \{ \tau(t)+1, \tau^* \}, & \text{if } \B{u}(t) = \B{u}(t-1), \\
        1, &\text{otherwise}.
    \end{cases}  
\end{equation}
Moreover, to simplify computations, when $\tau(t) = \tau^*$, if the MPC finds that $\B{u}(t)$ can be applied continuously for $m$ days, and that $\B{u}(t) = \B{u}(t-1)$, we apply $\B{u}(t)$ for $\min \{ m, l_{\R{MPC}} \}$ days (rather than just one day), for some $l_{\R{MPC}} \in \BB{N}_{>0}$.

The optimization problems are solved using a genetic algorithm implemented in \textsc{Matlab} 2020 \cite{Ga_toolbox}, with 
a population size of 4000, elite percentage of 20\%, crossover fraction of 40\%, and stall generations equal to 50.
The values we used for the control parameters are:
$T = 365$,
$\tau^* = 14$,
$T_\R{p} = 29$,
$l_{\R{MPC}} = 5$,
$\epsilon_H = 0.3$,
$\epsilon_\C{R} = 1.3$,
$d = 0.9$,
$c = 0.99$.

\subsection{Numerical results}%
\label{sec:numerical_results}




\subsubsection{Epidemic suppression}
We can observe the solution to \eqref{eq:optimization_problem_suppression} (with constraint \eqref{eq:optimization_problem_suppression_stability}) in Figure \ref{fig:scen1}.
As expected, the epidemic is effectively eradicated by the control strategy, as the number of infected people asymptotically goes to zero.
However, this strategy is particularly expensive with a total cost of \euro $421.989 \cdot 10^9$ (computed with $d = 1$).

\begin{figure}[t]
    \centering
    \includegraphics[trim=0.9cm 6.3cm 0.8cm 6.5cm, clip,width=\textwidth]{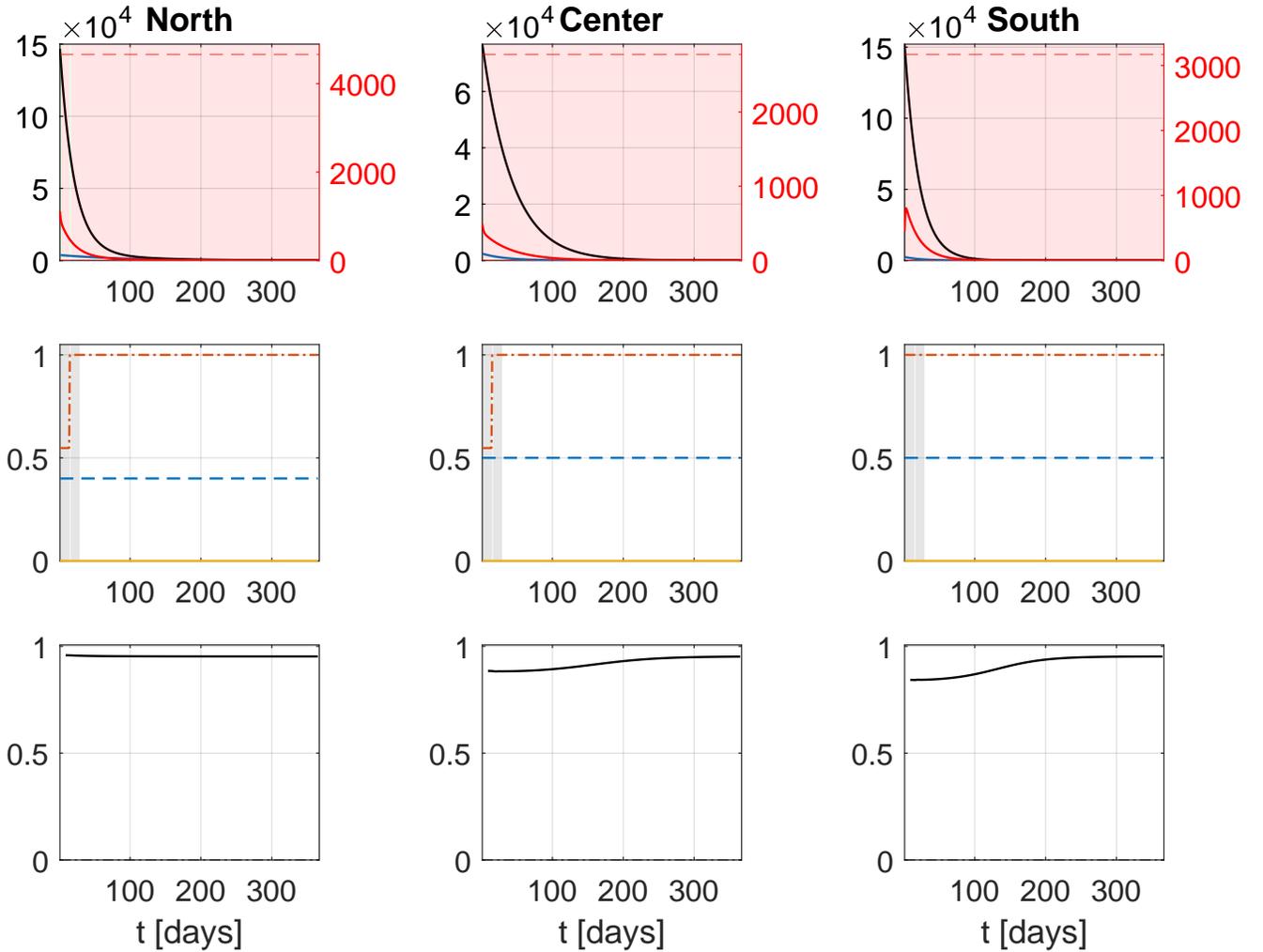}
    \caption{Evolution in each region of the model states (top panels), the control inputs (middle panels) and the values of $\mathcal{R}_{t,i}$ (bottom panels) when MPC is used to solve problem \eqref{eq:optimization_problem_suppression} (with constraint \eqref{eq:optimization_problem_suppression_stability}). 
    In the top panels, the scale on the left vertical axis refers to $I_i$ (blue) and $Q_i$ (black), whereas the scale in red on the right vertical axis refers to $0.1 H_i$ (solid red), $T_i^H$ (dashed red), and $\epsilon_H T_i^H$ (dash-dotted red). 
    Plots are shaded in red when condition \eqref{eq:stability_condition} is applied.
In the middle panels, the control variables $\rho_i$ (blue), $\varphi_i$ (orange), and $\sigma_i$ (yellow) are depicted, with the plots being shaded in grey when $\tau(t) < \tau^*$ (and the control values cannot be changed). 
In the bottom panels, the current estimate of $\C{R}_{t,i}$ estimated using \eqref{eq:estimation_R_t_i} is plotted.}
    \label{fig:scen1}
\end{figure}

\subsubsection{Mitigation}

The results of the control strategy when applied to solve \eqref{eq:optimization_problem_suppression} with constraint \eqref{eq:relaxed_stability} is reported in Figure \ref{fig:scen2}. 
In this case, the total cost is \euro $337.172\cdot 10^9$, significantly lower than the previous scenario.
This is because in each region NPIs interventions are taken only when some safety thresholds are met and only in that specific region.
At the same time, the number of people in ICUs remains within capacity, and all regional effective reproduction numbers $\C{R}_{t,i}$ remain within acceptable values.

\begin{figure}[t]
    \centering
    \includegraphics[trim=1.1cm 6.3cm 0.8cm 6.5cm, clip,width=\textwidth]{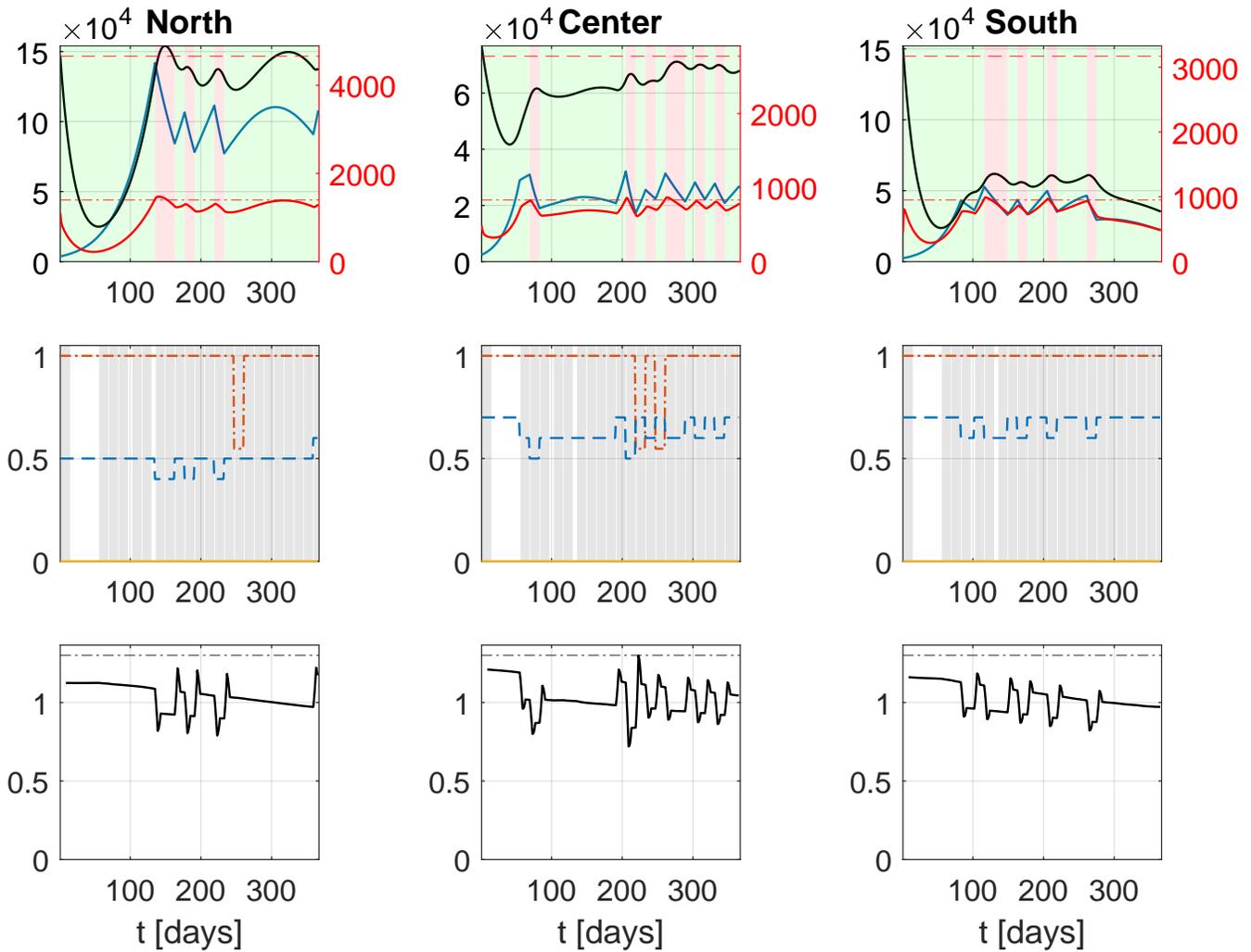}
    \caption{Evolution in each region of the model states (top panels), the control inputs (middle panels) and the values of $\mathcal{R}_{t,i}$ (bottom panels) when MPC is used to solve \eqref{eq:optimization_problem_suppression} with constraint \eqref{eq:relaxed_stability}. 
    Top panels are shaded in red when $A_i(\B{x}(t), \B{u}(t))\le c$ (see \eqref{eq:A_t_i}), while green is used when $A_i(\B{x}(t), \B{u}(t)) > c$.
    In the bottom panels, the threshold value $\epsilon_\C{R}$ for $\C{R}_{t,i}$ is plotted in a dash-dotted line.}
    \label{fig:scen2}
\end{figure}

\subsubsection{Effect of additional testing}

Finally, we report the solution to the problem when $\sigma_i(t) \in \C{S}_\sigma$ in Figure \ref{fig:scen3}.
The total cost is \euro $262.379\cdot 10^9$, even lower than the previous case.
Clearly, the use of additional testing allows to further reduce costs by relying less on social distancing to mitigate the epidemic.
Indeed, in Figure \ref{fig:scen3}, it is possible to note that travel restrictions are never used in this case, and social distancing is significantly less severe when compared to that in Figures \ref{fig:scen1} and \ref{fig:scen2}.

\begin{figure}[t]
    \centering
    \includegraphics[trim=1.1cm 6.3cm 0.8cm 6.5cm, clip,width=\textwidth]{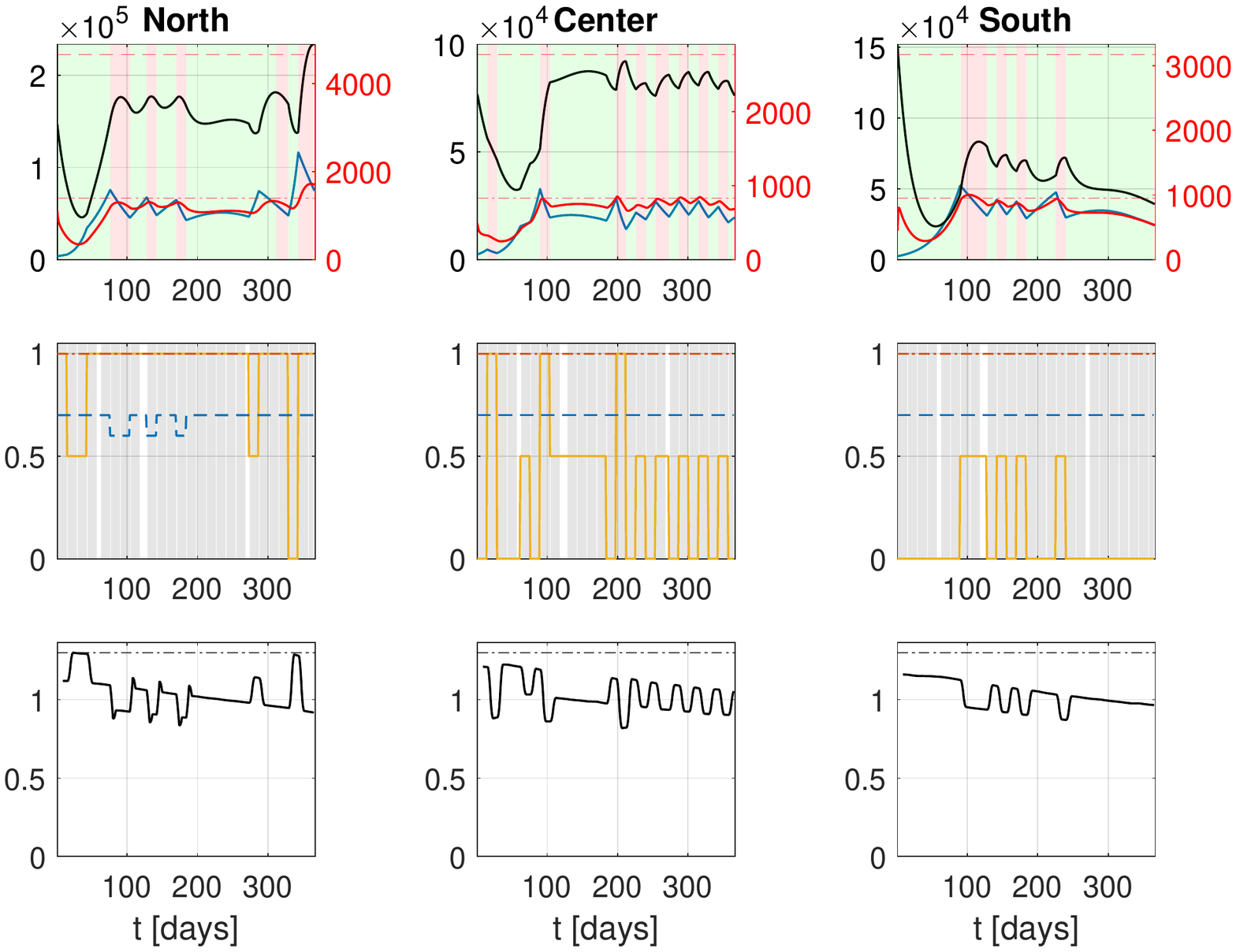}
    \caption{Evolution in each region of the model states (top panels), the control inputs (middle panels) and the values of $\mathcal{R}_{t,i}$ (bottom panels) when MPC is used to solve the problem when $\sigma_i \in \C{S}_\sigma$.}
    \label{fig:scen3}
\end{figure}

\section{Conclusions}%
\label{sec:conclusion}%
We considered the problem of designing localized intermittent non-pharmaceutical intervention strategies to mitigate the COVID-19 epidemic in a regional model of the Italian epidemic. 
We first gave a sufficient condition to determine the stability of the infected population, in a network epidemic model.
We then exploited this condition, or a relaxed version of it, as a constraint to formulate the control problem as an optimization problem in which the levels of social distancing, travel restrictions and additional testing must be selected to avoid saturation of the regional health systems and containment (or suppression) of the epidemic spread, while minimizing the overall cost to the economy.
Our results show the effectiveness of the proposed approach.
In the future, we will study robustness of our control strategy to model uncertainties, and measurement error.


\section*{Code and data availability}%
The code and data to run the model identification (§ \ref{sec:model_description}), the linear regression (§ \ref{sec:model_description}), and the simulations (§ \ref{sec:numerical_results}) are available at https://github.com/diBernardoGroup/MPC-to-mitigate-COVID-19.


\bibliographystyle{IEEEtran}

\appendix
For the sake of completeness, we explore here the relationship between the stability condition in Proposition \ref{pro:stability_trajectory} and a recent condition derived independently in \cite{ma2021optimal}, when both are applied to a discrete-time SIS model (adapted from the continuous-time model adopted in \cite{ma2021optimal}).

Namely, the dynamics of the infected we consider is given by
\begin{align}\label{eq:SIS_model_discrete}
    \B{I}(t+1) = \big[ (1-\gamma){\C{I}} + \text{diag}(\B{1} - \B{I}(t))\beta \B{A} \big] \B{I}(t),
\end{align}
where $I_i \in [0, 1]$ is the portion of infected in region $i$, with $i \in \{1, \dots, M\}$,
$\B{I} \coloneqq [I_1 \ \cdots \ I_M]\T$, $0\le \beta \le 1$ is the infection rate, $0 \le \gamma \le 1$ is the recovery rate, $\B{A} \in \BB{R}^{M\times M}$ captures the rate at which infection flows between regions (also accounting for social distancing), with $\B{A} \succeq 0$, and  $\R{diag}(\B{v})$ is a diagonal matrix having on its main diagonal the vector $\B{v}$.

Next, we derive and compare two different upper bounds on the maximum eigenvalue of $\mathbf{A}$, that guarantee the stability of the infected dynamics in the origin.
The first bound, say (B1), is derived following the arguments in Proposition \ref{pro:stability_trajectory}, whereas the second, say (B2), is computed following the steps in \cite{ma2021optimal}.

\subsubsection*{Bound derived from Proposition \ref{pro:stability_trajectory}}

By using similar arguments as those used in Proposition \ref{pro:stability_trajectory}, it can be shown that, if there exists some matrix norm and some $c<1$ such that
\begin{equation}\label{eq:proof_step_03}
    \norm{(1 - \gamma){\C{I}} + \R{diag}(\B{1} - \B{I}(t)) \beta \B{A}} \le c, 
    \quad \forall t \in \BB{N}_{\ge 0},
\end{equation}
then $\norm{{\B{I}}(t)} \le c^{t-t_0} \norm{{\B{I}}(t_0)}$ for all $t \ge t_0$.
Letting $I^\R{m}(t) \coloneqq \min_i I_i(t)$, and applying the triangle inequality and Lemma \ref{lem:bounds_matrices_norms}, we find that  \eqref{eq:proof_step_03} is satisfied if 
\begin{equation}\label{eq:proof_step_02}
    \norm{\B{A}} \le \frac{1}{1 - I^\R{m}(t)}\frac{\gamma+c-1}{\beta}, \quad \forall t \in \BB{N}_{\ge 0},
\end{equation}
where we used the fact that the variables $I_i(t)$ are bounded between $0$ and $1$. 
As $\norm{\B{A}} \ge \lambda_{\max}({\B{A}})$, then there exists $p \in (0,1]$ such that $p \norm{\B{A}} = \lambda_{\max}(\B{A})$.
Thus, we rewrite \eqref{eq:proof_step_02} and obtain the first bound as
\begin{equation}\tag{B1}\label{eq:our_condition_sis}
    \lambda_{\max}(\B{A}) \le b_1 \coloneqq \frac{p}{1 - I^\R{m}(t)} \frac{\gamma + c - 1}{\beta}.    
\end{equation}

\subsubsection*{Bound derived from \cite{ma2021optimal}}

Following the arguments in \cite{ma2021optimal}, we let $\B{P} \coloneqq (1-\gamma){\C{I}} + \beta \B{A}$, noting that $\B{P} \succeq 0$.
Applying the Perron-Frobenius theorem \cite[Theorem~8.3.1]{horn2012matrix} to $\mathbf{P}\T$ yields the existence of an eigenvalue $\lambda_{\max}(\B{P})\ge 0$, with associated left eigenvector $\B{v}_{\max} \succeq 0$. 
Moreover,  
$\lambda_{\max}(\B{P})$ is such that $\lambda_{\max}(\B{P}) \ge |\lambda_i(\B{P})|$, for all other eigenvalues $\lambda_i$ of $\B{P}$.
Now, if
\begin{equation}\tag{B2}\label{eq:their_condition_sis}
    \lambda_{\R{max}}(\B{A}) \le b_2 \coloneqq \frac{\gamma + c -1}{\beta},
\end{equation}
(which constitutes the second bound), we consequently get $\lambda_{\R{max}}(\B{P}) \le c < 1$, and it is then immediate to show that this implies that $\B{0}$ is a stable equilibrium for the dynamics $\B{I}(t)$ given in (\ref{eq:SIS_model_discrete}).
Indeed, we have
\begin{equation*}
\begin{split}
    \B{v}_{\max}^{\T}\B{I}(t+1) &= \B{v}_{\max}^{\T}\big((1-\gamma){\C{I}} + \R{diag}(\B{1}-\B{I}(t))\beta \B{A}\big)\B{I}(t)\\
    &\le \B{v}_{\max}^{\T} \B{P}\B{I}(t)
    = \lambda_{\max}(\B{P}) \B{v}_{\max}^{\T} \B{I}(t)
    \le c \B{v}_{\max}^{\T} \B{I}(t),
    \end{split}
\end{equation*}
which means stability of $\B{I}(t)$ in $\B{0}$ recalling that $\B{I}(t) \succeq 0$ and $\B{v}_{\max} \succeq 0$.
Note that \eqref{eq:their_condition_sis} is equivalent to the stability condition in Proposition 3.1 in \cite{ma2021optimal} when applied to a discrete-time SIS network model.

\subsubsection*{Comparison}

Finally, we relate the two stability conditions \eqref{eq:our_condition_sis} and \eqref{eq:their_condition_sis}.
Clearly, when $I^\R{m}(t) \in (1-p, 1)$, $b_1 > b_2$ and our condition \eqref{eq:our_condition_sis} is less conservative than \eqref{eq:their_condition_sis}, while it is more conservative ($b_1 < b_2$) when $I^\R{m}(t) \in [0, 1-p)$.

\end{document}